\documentclass[journal,twoside,web]{ieeecolor}
\usepackage{lcsys}
\usepackage{cite}
\usepackage{latexsym}
\usepackage{amsfonts,amssymb,amsmath,graphicx, bbold}
\usepackage{algorithmic}
\usepackage{graphicx}
\usepackage{textcomp}
\usepackage{etoolbox}
\usepackage{array}

\newtheorem{theorem}{Theorem}
\newtheorem{lemma}[theorem]{Lemma}

\newcounter{constraint}[section]
\newenvironment{constraint}[1][]{\refstepcounter{constraint}\par\medskip
    \textit{Assumption~\theconstraint #1} \rmfamily}{\medskip}

\newcommand{\rvec}[1]{{\mathbf{\MakeLowercase{#1}}}}

\newcommand{\rs}[1]{\mathbf{\MakeLowercase{#1}}}
\newcommand{\dvec}[1]{{#1}}

\def\BibTeX{{\rm B\kern-.05em{\sc i\kern-.025em b}\kern-.08em
    T\kern-.1667em\lower.7ex\hbox{E}\kern-.125emX}}
\begin{document}
\title{A Lower-bound for Variable-length Source Coding in Linear-Quadratic-Gaussian Control with Shared Randomness}
\author{Travis C. Cuvelier, \IEEEmembership{Student Member, IEEE}, Takashi Tanaka,  \IEEEmembership{Senior Member, IEEE}, and Robert W. Heath, Jr., \IEEEmembership{Fellow, IEEE}
\thanks{First submission on March 12, 2022. This work was supported in part by the National Science Foundation under Grant No. ECCS-1711702 and and CAREER Award \#1944318. }
\thanks{T. Cuvelier is with the Department
of Electrical and Computer Engineering, The University of Texas at Austin,
TX, 78712 USA e-mail: tcuvelier@utexas.edu.}
\thanks{T. Tanaka is with the Department
of Aerospace Engineering and Engineering Mechanics, The University of Texas at Austin,
TX, 78712 USA e-mail: ttanaka@utexas.edu.}
\thanks{R. Heath is with the Department
of Electrical and Computer Engineering, North Carolina State University, Raleigh, 
NC, 27606 USA e-mail: rwheath2@ncsu.edu.}}

\maketitle
\thispagestyle{empty}
\begin{abstract}
In this letter, we consider a Linear Quadratic Gaussian (LQG) control system where feedback occurs over a noiseless binary channel and derive lower bounds on the minimum communication cost (quantified via the channel bitrate) required to attain a given control performance. We assume that at every time step an encoder can convey a packet containing a variable number of bits over the channel to a decoder at the controller. Our system model provides for the possibility that the encoder and decoder have shared randomness, as is the case in systems using dithered quantizers. We define two extremal prefix-free requirements that may be imposed on the message packets; such constraints are useful in that they allow the decoder, and potentially other agents to uniquely identify the end of a transmission in an online fashion. We then derive a lower bound on the rate of prefix-free coding in terms of directed information; in particular we show that a previously known bound still holds in the case with shared randomness. We generalize the bound for when prefix constraints are relaxed, and conclude with a rate-distortion formulation. 
\end{abstract}

\section{Introduction}
\label{sec:introduction}
\IEEEPARstart{I}{n} this letter, we derive and analyze lower bounds on the minimum bitrate of feedback communication required to obtain a given LQG control performance. We consider the setting of variable-length coding; at each discrete time instant an encoder that can fully observe the plant conveys a packet containing a variable number of bits to a decoder co-located with the controller. Various prefix constraints can be imposed on the packets; these allow the decoder, and perhaps other agents, to uniquely identify the end of each codeword given varying degrees of common knowledge/side information (SI). Prefix constraints are useful when a communication medium is shared; upon detecting the end one user's codeword, other users can identify the channel as free-to-use.

We assume that the packets are conveyed over a noiseless, error-free communication channel. The bounds we derive for noiseless channels are useful even when considering real-world noisy channels. They can be compared directly with an appropriate notion of the ``real" channel's capacity. If our lower bound exceeds this capacity, the desired control performance is not achievable. Likewise, if source-channel separation is imposed on the design of the communication architecture, our lower bounds apply to the optimal source codec (encoder/decoder pair) design.

We allow for the possibility that the encoder and decoder have access to \textit{shared randomness}; namely an IID sequence of exogenous random variables that are revealed causally to both the encoder and decoder. Shared randomness of this nature arises in settings where the encoder and decoder use dithered quantization. In dithered quantization, randomness is intentionally introduced into the quantization process to make the quantization error more amenable to analysis. In particular, dithered quantization has been used to design schemes for minimum bitrate LQG control (cf. \cite{silvaFirst}\cite{tanakaISIT}). While these achievability approaches are often compared to known lower bounds that apply without shared randomness, in this letter we clarify that these lower bounds do not change even when a dither signal is available.

Recently, \cite{milanCorrection} demonstrated that the proofs of the lower bounds to which the achievability approaches in \cite{silvaFirst} and \cite{tanakaISIT} were compared are invalid when the encoder and decoder share randomness. The proof of the lower bound in \cite{silvaFirst} is corrected in \cite{milanCorrection}. In this letter, we formalize two extremal prefix constraints. These notions have been conflated in the prior literature. We correct the proof of the lower bound in \cite{tanakaISIT} using a novel proof technique. We show that for all the prefix constraints we consider and irrespective of the marginal distribution of the shared randomness, the channel bitrate is lower bounded by the time-average directed information (DI) from the state vector to the control input. Notably, this is the same lower bound that applies without shared randomness \cite{SDP_DI}\cite{kostinaTradeoff}, demonstrating that systems employing dithered quantization are subject to the same fundamental lower bounds as systems without dithering \cite{SDP_DI}.  While our result is consistent with \cite{milanCorrection}, we believe our proof is simpler. 

The bitrate of lossless source coding can be reduced by relaxing prefix constraints \cite{verduVariableLength}. We show how the derived lower bounds change upon lifting prefix constraints and conclude with a rate distortion formulation, specialized to the case of time-invariant plants and an infinite horizon, following from \cite{SDP_DI}. As our lower bound is proved in the finite horizon, other relevant rate distortion formulations can be derived from, e.g. \cite{SDP_DI} and \cite{reviewerPaper}.  A brief literature review now follows.
\begin{figure}[t]
	\centering
	\includegraphics[scale = .12]{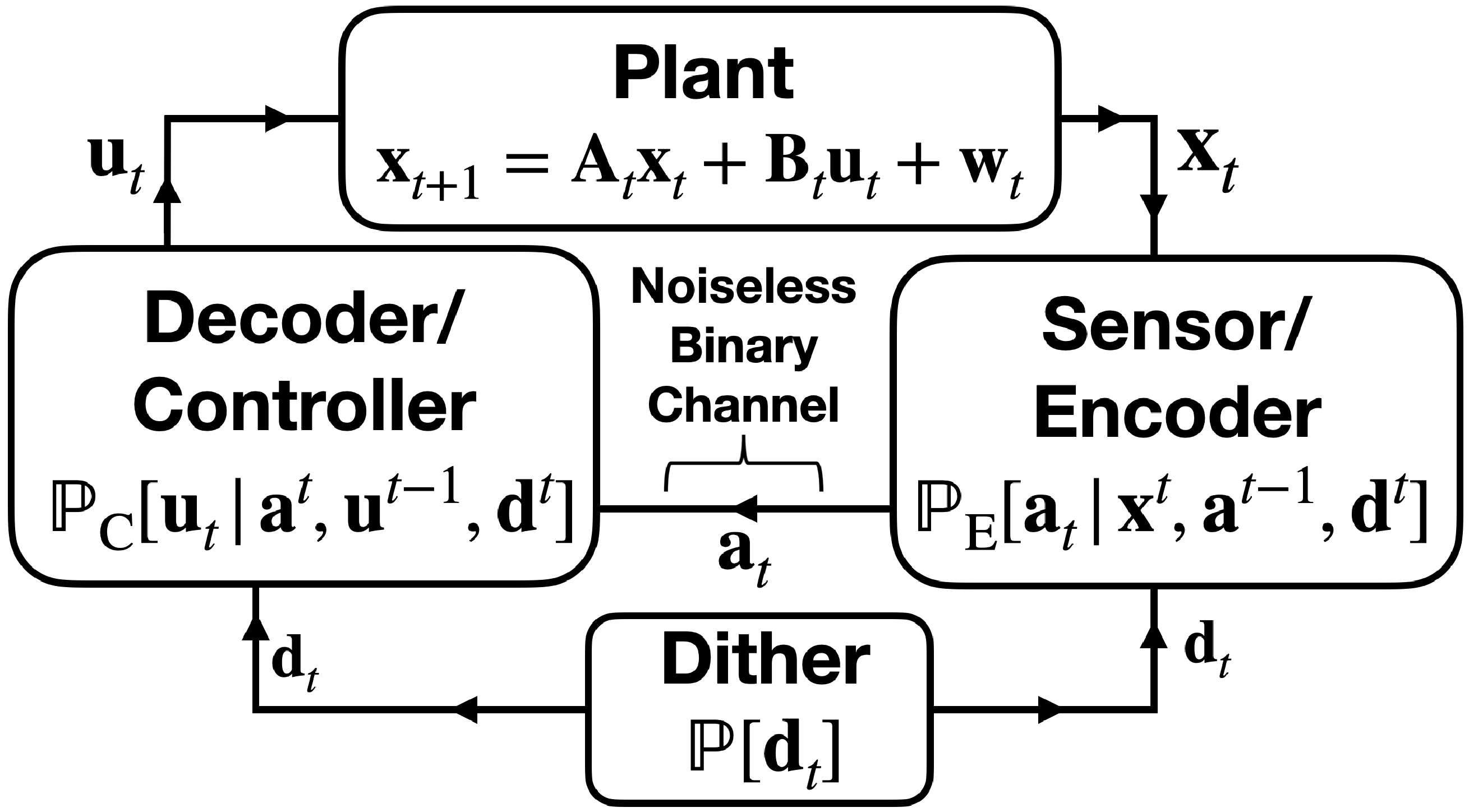}
	\caption{The factorization in (\ref{eq:exog}) states that both the shared dither and the process noise are drawn independently given all prior encoder and and control actions. The factorization in (\ref{eq:cwfact}) provides that the codeword at time $t$ can be drawn from some distribution that depends on its prior observations of the state and dither sequence, as well as its past actions. Likewise, the control input can be drawn randomly given its the decoder's observations of received codewords, the dither, and its past actions.}\label{fig:ditharch}
\end{figure}
\subsection{Literature Review}
Our work follows from the problem formulation of \cite{silvaFirst}, which considered a SISO LQG control system where feedback measurements were conveyed from an encoder to a decoder over a noiseless binary channel. In the variable length setting and enforcing a prefix constraint on the packets, \cite{silvaFirst} derived a lower bound for the time-average expected channel bitrate in terms of Massey's DI \cite{masseyDI}. Given a constraint on the LQG cost, \cite{silvaFirst} showed that the lower bound is nearly achievable when the encoder and decoder share access to a common dither signal. In \cite{SDP_DI} and \cite{tanakaISIT}, the work in \cite{silvaFirst} was extended to MIMO plants. In particular, \cite{SDP_DI} developed a rate-distortion formulation in terms of semidefinite programming; a semidefinite program (SDP) was derived to compute the tradeoff between the minimum DI and LQG cost. The achievability of the lower bound, again assuming dithering, was demonstrated in \cite{tanakaISIT}. Analytical lower bounds on the relevant DI as a function of the maximum tolerable LQG cost were developed in \cite{kostinaTradeoff}. It was also shown that the entropy rate of an innovations quantizer approaches this bound without the use of dithering. Notably, \cite{kostinaTradeoff} generalized the entropy lower-bound for prefix-free coding to the setting without prefix constraints (see also \cite{verduVariableLength}). 

We consider a setup where the encoder and decoder share randomness. This is in contrast to a setup where the decoder can access a more traditional notion of SI, namely a random variable correlated with the plant's state vector. The impact of this latter notion of SI on the communication/LQG cost tradeoff was investigated in \cite{kostinaSI}, \cite{photisSI}, \cite{oronNew},  and \cite{ourciss}. In particular, these works consider linear/Gaussian observations of the plant available at the decoder. Rate distortion formulations were considered in \cite{kostinaSI}, \cite{photisSI}, and \cite{oronNew}, meanwhile an achievability approach (assuming noiseless SI also available at the encoder) was given in \cite{ourciss}. We will show that shared randomness does not affect lower bounds on bitrate. 

The lower bound on bitrate derived in  \cite{silvaFirst} purported to apply to quantization/coding schemes with shared dither sequences at the encoder and decoder. A flaw in the proof of the bound from \cite{silvaFirst} was recently discovered by \cite{milanCorrection}. The proof was revised using new DI data processing inequalities derived in \cite{milanMainResult}. In our letter, while we prove a lower bound similar to that in \cite{milanCorrection}, our problem formulation and proof techniques differ significantly. The data processing inequalities in \cite{milanMainResult} apply to general feedback systems consisting of deterministic causal stages, where system blocks are randomized through exogenous inputs. We consider an alternative formulation, and prove that a data processing inequality holds under natural conditional independence assumptions between the system variables. Our lower bound then follows directly. 
\subsection{Our contributions}
In summary, the contributions of this letter are:
\begin{enumerate}
    \item we define two different prefix constraints that can be imposed on the feedback packets. In prior work, these constraints have been used somewhat ambiguously. Namely, we define a strict as well as a relaxed constraint and show that they are subject to the same lower bound. We highlight the operational significance of these constraints in control systems. 
    \item when the encoder and decoder share randomness, we derive a DI data processing inequality directly from the factorization of the joint distribution of the system variables. This inequality proves that for two extremal notions of what it means to be ``prefix-free", the DI lower bound from \cite{tanakaISIT} holds even when the encoder and decoder share randomness. Following from \cite{kostinaTradeoff} and \cite{verduVariableLength}, we generalize the bound to codecs without prefix constraints.
\end{enumerate}
 The lower bounds lead to a rate distortion formulation, which, following from \cite{SDP_DI}, can be written as an SDP.
\subsection{Notation}
We denote constant scalars and vectors in lowercase $x$, scalar and vector random variables in boldface $\rs{x}$, and matrices by capital letters $X$. The set of finite-length binary strings is denoted $\{0,1\}^*$. We use $H$ for the entropy of a discrete random variable and $I$ for mutual information (MI). For time domain sequences, let $\{\rs{x}_{t}\}$ denote $(\rs{x}_{0},\rs{x}_{1},\dots)$. We let $\rs{x}_{a}^{b}$ denote $(\rs{x}_{a},\dots,\rs{x}_{b})$ if $b\ge a$, and let $\rs{x}_{a}^{b}=\emptyset$ otherwise. Likewise, let $\rs{x}^{b}= \rs{x}_{0}^{b}$. Given $\{\rs{x}_{t}\}$ we define the shifted sequence $\{\rs{x}_{t}^{+}\}$ by $(0,\rs{x}_{0},\rs{x}_1,\dots)$. For $\{\rs{a}_{t},\rs{b}_{t},\rs{c}_{t}\}$, Massey's causally conditioned DI is defined by (cf. \cite{masseyDI})
\begin{align}\label{eq:didef}
    I(\rs{a}^{T-1}\rightarrow \rs{b}^{T-1}||\rs{c}^{T-1}) = \sum_{t=0}^{T-1}I(\rs{a}^{t};\rs{b}_{t}|\rs{b}^{t-1},\rs{c}^{t}).
\end{align}

\section{System Model and Problem Formulation}\label{sec:systemmodel}
The system model under consideration is depicted in Figure \ref{fig:ditharch}. We consider a general MIMO plant with a feedback model where communication takes place over a noiseless binary channel. We assume that a time-invariant plant is fully observable to an sensor/encoder block, which conveys a binary codeword $\rvec{a}_{t}\in\{0,1\}^*$ over the channel to a combined decoder/controller. Upon receipt of the codeword, the decoder/controller designs the control input. We denote the state vector as $\rvec{x}_t\in \mathbb{R}^{m}$ and the control input as $\rvec{u}_{t}\in\mathbb{R}^{u}$. We assume that the sensor/encoder and decoder/controller share access to a common random \textit{dither signal}, $\{\rs{d}_{t}\}$. The dither is assumed to be IID over time. Note that this system model includes systems where a dither is unavailable as a special case (e.g. we could always set $\rs{d}_{t}\overset{\mathrm{a.s.}}{=}0$ for all $t$). The process noise $\rvec{w}_{t}\sim\mathcal{N}(\dvec{0},{W})$ is assumed to be IID over time. We assume ${W}\succ{0}_{m\times m}$. We assume assume that $\rvec{x}_{0}\sim\mathcal{N}(\dvec{0},{X}_0)$ for some $X_0\succeq 0$, and that $\{\rs{w}_{t}\}$, $\{\rs{d}_{t}\}$, and $\rs{x}_{0}$ are mutually independent. Given a sequence of system matrices ${A}_{t}\in\mathbb{R}^{m\times m}$ and feedback gain matrices ${B}_{t}\in\mathbb{R}^{m\times u}$ for $t\ge 0$ the plant dynamics are given by
\begin{align}\label{eq:ssmodel}
    \rvec{x}_{t+1} = {A}_{t}\rvec{x}_{t}+{B}_{t}\rvec{u}_{t}+\rvec{w}_{t}. 
\end{align} We assume that the sensor/encoder and the decoder/controller are stochastic; i.e. that they can use randomized strategies.\footnote{E.g., the encoder can draw the codeword randomly given its input.} The sensor/encoder policy is assumed to be a sequence of causally conditioned kernels given by 
\begin{align}\label{eq:encdithpol}
    \mathbb{P}_{\mathrm{E}}[\rvec{a}^{\infty}|| \rvec{d}^{\infty},\rvec{x}^{\infty}] = \{ \mathbb{P}_{\mathrm{E}}[\rvec{a}_{t}|\rvec{a}^{t-1},\rvec{d}^{t},\rvec{x}^{t}]\text{ for }t\in\mathbb{N}_{0}\}. 
\end{align} Likewise, the corresponding decoder/controller policy is given by the sequence of causally conditioned kernels 
\begin{align}\label{eq:contdithpol}
    \mathbb{P}_{\mathrm{C}}[\rvec{u}^{\infty}|| \rvec{a}^{\infty},\rvec{d}^{\infty}] = \{ \mathbb{P}_{\mathrm{C}}[\rvec{u}_{t}|\rvec{a}^{t},\rvec{d}^{t},\rvec{u}^{t-1}]\text{ for }t\in\mathbb{N}_{0}\}.
\end{align} By (\ref{eq:ssmodel}), for all $t$, $\rs{x}^{t}$ is a deterministic function of $\rvec{x}_{0}$, $\rvec{u}^{t-1}$, and $\rvec{w}^{t-1}$. We assume that the one-step transition kernels between  $\rvec{a}_{t}$, $\rvec{d}_{t}$, $\rvec{u}_{t}$, and  $\rvec{w}_{t}$, factorize for all $t\ge0$ as \begin{subequations}\label{eq:ditherFactorization}
\begin{multline}\label{eq:cwfact}
      \mathbb{P}[\rvec{a}_{t+1},\rvec{u}_{t+1}|\rvec{a}^{t},\rvec{d}^{t+1},\rvec{u}^{t},\rvec{w}^{t},\rvec{x}_{0}] =\\\mathbb{P}_{\mathrm{E}}[\rvec{a}_{t+1}|\rvec{a}^{t},\rvec{d}^{t+1},\rvec{x}^{t+1}]\mathbb{P}_{\mathrm{C}}[\rvec{u}_{t+1}|\rvec{a}^{t+1},\rvec{d}^{t+1},\rvec{u}^{t}]\text{, and }
\end{multline}  
\begin{multline}\label{eq:exog} \mathbb{P}[\rvec{a}_{t+1},\rvec{d}_{t+1},\rvec{u}_{t+1},\rvec{w}_{t+1}|\rvec{a}^{t},\rvec{d}^{t},\rvec{u}^{t},\rvec{w}^{t},\rvec{x}_{0}] = \\\mathbb{P}[\rvec{a}_{t+1},\rvec{u}_{t+1}|\rvec{a}^{t},\rvec{d}^{t},\rvec{u}^{t},\rvec{w}^{t},\rvec{x}_{0}]\mathbb{P}[\rvec{d}_{t+1}]\mathbb{P}[\rvec{w}_{t+1}].
\end{multline}
 \end{subequations}
Implications of these factorization are discussed in Fig. \ref{fig:ditharch}. 
 
 The length of the binary codewords $\{\rs{a}_{t}\}$ provides a notion of communication cost. This is motivated by a scenario where measurements from a remote sensor platform are conveyed over wireless to control a plant. In general, minimizing the necessary bitrate from the remote platform to the controller minimizes the amount of physical layer resources that must be allocated to the particular link. The problem of interest is to minimize this bitrate subject to a constraint on the LQG control performance. In this work, we are concerned primarily with deriving lower bounds on the bitrate. 
 \begin{figure*}[t]
	\centering
	\includegraphics[scale = .2515]{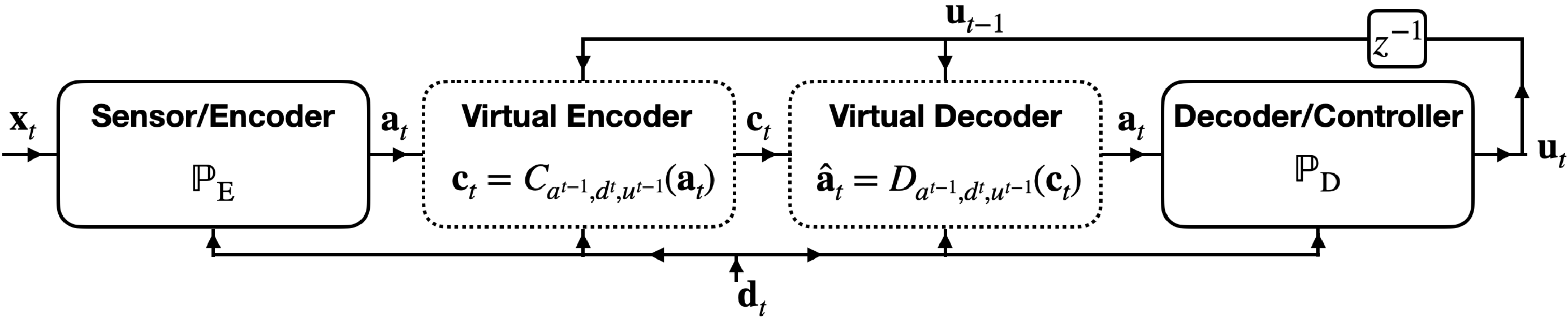}
    \vspace{-.35cm}
	\caption{The virtual encoder produces the codeword $\rvec{c}_{t}$ given access to $\rs{a}_{t}$ and realizations of $(\rs{a}^{t-1},\rs{d}^{t},\rs{u}^{t-1})$. The virtual codec is \textit{lossless}; we require that given $\rs{c}_{t}$ and the realizations of $(\rs{a}^{t-1},\rs{d}^{t},\rs{u}^{t-1})$, the virtual encoder reconstruct $\rs{a}_{t}$ exactly, ensuring equivalent control performance in the original system. The virtual codewords $\rs{c}_{t}$ must also satisfy Assumption \ref{cs:constraint2}. Notably, the virtual encoder has access to more side information than the sensor/encoder (cf. Fig. \ref{fig:ditharch}).    } \label{fig:virtualcodec}
\end{figure*} 
 At every time $t$, we require $\rvec{a}_{t}$ to satisfy a prefix constraint. This allows the decoder (and possibly other agents sharing the same communication network) to uniquely identify the end of the transmission from the encoder. For $\dvec{a}\in\{0,1\}^*$, let $\ell(\dvec{a})$ denote the length of $\dvec{a}$ in bits. The prefix constraint allows us to derive simple lower-bounds on $\mathbb{E}[\ell(\rvec{a}_{t})]$. We are interested in the optimization
\begin{equation}\label{eq:codewordLenghtOptimization}
\begin{aligned}
& \underset{\mathbb{P}_\mathrm{E}, \mathbb{P}_{\mathrm{C}}}{\inf}  \text{ }\frac{1}{T+1}\sum\nolimits_{t=0}^{T}\mathbb{E}[\ell(\rvec{a}_{t})] \\ &\text{s.t. }   \frac{1}{T+1}\sum\nolimits_{t=0}^{T}\mathbb{E}[\lVert \rvec{x}_{t+1} \rVert_{{Q}}^{2} +\lVert \rvec{u}_{t} \rVert_{{R}}^{2}] \le \gamma,
\end{aligned} 
\end{equation} where ${Q}\succeq {0}$, ${R}\succ {0}$, and $\gamma$ is the maximum tolerable LQG cost. The minimization is over admissible sensor/encoder and decoder/controller policies described by 
(\ref{eq:encdithpol}) and (\ref{eq:contdithpol}). In \cite{SDP_DI}, it was shown that, for policies without additive dithering, (\ref{eq:codewordLenghtOptimization}) is lower-bounded by an SDP (technically a log-determinant optimization) where a particular DI is minimized over the space of linear/Gaussian policies. In the sequel, we show that this lower bound still holds for architectures with additive dither satisfying (\ref{eq:encdithpol})-(\ref{eq:ditherFactorization}). 
\section{Lower bounds}
We now derive lower bounds on the rate of prefix-free coding within the feedback loops of Fig. \ref{fig:ditharch}. These bounds apply irrespective of the marginal distribution of the dither signal $\rs{d}_{t}$, and thus also apply to systems without dithering.

\subsection{Directed information lower bound}\label{ssec:dilb} We first formally define various prefix constraints that may be imposed on the codeword $\rs{a}_t$. We consider two distinct notions, and demonstrate that the same lower bound applies to both. We first require that, at every time $t$, $\rs{a}_t$ is a codeword from a prefix-free code that can be decoded by any decoder with knowledge of the marginal distribution $\mathbb{P}_{\rs{a}_t}$. Assumption \ref{cs:constraint1} formalizes this constraint. 

\begin{constraint}\label{cs:constraint1}
    For all distinct $a_1,a_2\in\{0,1\}^*$ with \newline $\mathbb{P}_{\rs{a}_t}[\rs{a}_t=a_1]>0$ and $\mathbb{P}_{\rs{a}_t}[\rs{a}_t=a_2]>0$, $a_1$ is not a prefix of $a_2$ and vice-versa.
\end{constraint} 

Assumption \ref{cs:constraint1} was used, implicitly, as the notion ``prefix-free" in proofs of \textit{lower bounds} in \cite{silvaFirst} and \cite{tanakaISIT}. Assumption \ref{cs:constraint1} ensures that the decoder can uniquely identify the end of the codeword \textit{without} relying on its knowledge of the previously received codewords $\rs{a}^{t-1}$, its previously designed control inputs $\rs{u}^{t-1}$, and the common randomness $\rs{d}^{t}$ it shares with the encoder. However, Assumption \ref{cs:constraint1} is perhaps too restrictive; at time $t$ both the encoder and the decoder have access to some common knowledge, including $\rs{a}^{t-1}$ and $\rs{d}^{t}$. While information known \textit{only} to the encoder cannot reduce the minimum codeword length, information known to the decoder can. We consider prefix-free codes that are instantaneous with respect to realizations of the random variables known to the decoder in Assumption \ref{cs:constraint2}.
\begin{constraint}\label{cs:constraint2}
    For any realizations $(\rs{a}^{t-1},\rs{d}^{t},\rs{u}^{t-1})=(a^{t-1},d^{t},u^{t-1})$, for all distinct $a_1,a_2\in\{0,1\}^*$ with $\mathbb{P}_{\rs{a}_t|\rs{a}^{t-1},\rs{d}^{t},\rs{u}^{t-1}}[\rs{a}_t=a_1|(\rs{a}^{t-1},\rs{d}^{t},\rs{u}^{t-1})=(a^{t-1},d^{t},u^{t-1})]>0$ and $\mathbb{P}_{\rs{a}_t|\rs{a}^{t-1},\rs{d}^{t},\rs{u}^{t-1}}[\rs{a}_t=a_2|(\rs{a}^{t-1},\rs{d}^{t},\rs{u}^{t-1})=(a^{t-1},d^{t},u^{t-1})]>0$, $a_1$ is not a prefix of $a_2$ and vice-versa.
\end{constraint} 

This requirement ensures that given the knowledge of $\rs{a}^{t-1}$, $\rs{d}^{t}$, and $\rs{u}^{t-1}$ the decoder can uniquely identify the end of the codeword. While Assumption \ref{cs:constraint2} is somewhat relaxed in comparison to Assumption \ref{cs:constraint1}, the implementation of a coding scheme under Assumption \ref{cs:constraint2} is likely more cumbersome with respect to that of Assumption \ref{cs:constraint1}.  Under Assumption \ref{cs:constraint1}, a decoder can detect the end of the codewords without considering realizations of $(\rs{a}^{t-1}$, $\rs{d}^{t},\rs{u}^{t-1})$, whereas under Assumption \ref{cs:constraint2} this is not necessarily the case. In fact, an implication of Assumption \ref{cs:constraint2} is that the encoder and decoder may need to agree on infinitely many codebooks. The notion ``prefix-free" used in the \textit{achievability} architectures of \cite{silvaFirst} and \cite{tanakaISIT} do not satisfy Assumption \ref{cs:constraint1}, but do satisfy Assumption \ref{cs:constraint2}. In \cite{silvaFirst} and \cite{tanakaISIT}, the codewords are prefix-free \textit{given} the shared dither signal. The following lower bound applies under either Assumption \ref{cs:constraint1} or \ref{cs:constraint2}.
\begin{theorem}\label{thm:converse}
   In a system conforming to  Fig. \ref{fig:ditharch} and (\ref{eq:ditherFactorization}) with fixed encoder and decoder policies such that either Assumption \ref{cs:constraint1} or \ref{cs:constraint2} is satisfied at every $t$, the time-average expected codeword length satisfies
 \begin{align} \label{eq:converse_thm}   \frac{1}{(T+1)}\sum_{i=0}^{T}\mathbb{E}[\ell(\rs{a}_t)] \ge \frac{1}{(T+1)}I(\rs{x}^T\rightarrow\rs{u}^{T}).
 \end{align}
\end{theorem}
Theorem \ref{thm:converse} follows from Lemmas \ref{lemm:milb} and \ref{lemm:didpi}, stated presently.
\begin{lemma}\label{lemm:milb}
If $\rs{a}_t$ satisfies either of the prefix-free conditions outlined in Assumptions \ref{cs:constraint1} or \ref{cs:constraint2}, we have
\begin{align}\label{eq:lb_lemm}
    \frac{1}{(T+1)}\sum_{i=0}^{T}\mathbb{E}[\ell(\rs{a}_t)] \ge \frac{1}{(T+1)}I(\rs{x}^T\rightarrow\rs{a}^{T}||\rs{d}^{T},\rs{u}_{+}^{T}).\linebreak
\end{align}
\end{lemma}
\begin{proof}
We first derive a bound on the codeword length under Assumption \ref{cs:constraint1}. Assume that the encoder and decoder policies are fixed and conform to Assumption \ref{cs:constraint1}. Consider $\rs{a}_t$ itself as an information source with a range in $\{0,1\}^*$. Define the identity map  $C^{\mathrm{I}}: \{0,1\}^*\rightarrow\{0,1\}^*$ such that for all $a\in\{0,1\}^*$, $C^{\mathrm{I}}(a) = a$. By Assumption \ref{cs:constraint1}, at every time $t$, $C^{\mathrm{I}}$ is a lossless, prefix-free source code for $\rs{a}_t$. Thus, we have (cf. \cite[Theorem 5.3.1]{elemIT})
\begin{IEEEeqnarray}{rCl}\label{eq:converse1}
\mathbb{E}[\ell(\rs{a}_t)]&=&\mathbb{E}[\ell(C^{\mathrm{I}}(\rs{a}_t))] \label{eq:converse1ident}\\&\ge &H(\rs{a}_t)\label{eq:converse1shannon} \\ &\ge& H(\rs{a}_t|\rs{a}^{t-1}, \rs{d}^{t}, \rs{u}^{t-1},)\label{eq:converse1condred} \\ &\ge&I(\rs{a}_t;\rs{x}^{t}|\rs{a}^{t-1},\rs{d}^{t},\rs{u}^{t-1})\label{eq:defmi1}
\end{IEEEeqnarray} where (\ref{eq:converse1ident}) follows since $C^{\mathrm{I}}$ is an identify map, (\ref{eq:converse1shannon}) follows from the fact that $C^{\mathrm{I}}$ is a prefix-free code for $\rs{a}_t$ and thus has an expected length lower bounded by the entropy of $\rs{a}_t$. Equation (\ref{eq:converse1condred}) follows from the fact that conditioning reduces entropy. Finally, (\ref{eq:defmi1}) follows from subtracting the (nonnegative) discrete conditional entropy  $H(\rs{a}_t|\rs{a}^{t-1}, \rs{d}^{t}, \rs{u}^{t-1}, \rs{x}^{t})$ from (\ref{eq:converse1condred}) and applying the definition of MI. 

The logic leading to (\ref{eq:converse1shannon}) does not hold if Assumption \ref{cs:constraint1} is relaxed to Assumption \ref{cs:constraint2}. However, under Assumption \ref{cs:constraint2}, the lower bound (\ref{eq:defmi1}) still holds. To show this, we analyze system model that is ``more-generous" than Fig. \ref{fig:ditharch}. First, assume that the sensor/encoder and decoder/controller policies are fixed such that the $\rs{a}_{t}$ satisfy Assumption \ref{cs:constraint2} at every $t$.  With the sensor/encoder and decoder/controller policies fixed, consider modifying the system model of Fig. \ref{fig:ditharch} by inserting a second ``virtual"  (hypothetical), deterministic lossless source codec between the original encoder/sensor and decoder/controller. Fig. \ref{fig:virtualcodec} gives an overview of this modification. Note that, by Fig. \ref{fig:ditharch}, the original sensor/encoder does not have access to $\rs{u}^{t-1}$, however the virtual encoder does. At every time $t$, we allow the virtual encoder to produce the codeword $\rs{c}_{t}$ given $\rs{a}_{t}$ and realizations of $(\rs{a}^{t-1},\rs{u}^{t-1},\rs{d}^{t})$. We require that, given realizations of $(\rs{a}^{t-1},\rs{u}^{t-1},\rs{d}^{t})$ and the codeword $\rs{c}_{t}$, the virtual decoder reproduce $\rs{a}_{t}$ exactly. We require that the codewords $\{\rs{c}_{t}\}$ also satisfy Assumption \ref{cs:constraint2} (replacing $\rs{a}_t$ with the virtual codeword $\rs{c}_t$). 

Fix the realizations $(\rs{a}^{t-1},\rs{d}^{t},\rs{u}^{t-1})=({a}^{t-1},{d}^{t},{u}^{t-1})$ and consider encoding  $\rs{a}_{t}$ into $\rs{c}_{t}$. Define the virtual encoder function mapping $\rs{a}_{t}$ to $\rs{c}_{t}$ given the realizations by $C_{a^{t-1}, d^{t}, u^{t-1}}:\{0,1\}^{*}\rightarrow \{0,1\}^{*}$. Define the decoder function $D_{a^{t-1}, d^{t}, u^{t-1}}:\{0,1\}^{*}\rightarrow \{0,1\}^{*}$. By assumption, $D_{a^{t-1}, d^{t}, u^{t-1}}(C_{a^{t-1}, d^{t}, u^{t-1}}(\rs{a}_{t}))=\rs{a}_{t}$. Assume $C_{a^{t-1}, d^{t}, u^{t-1}}$ is chosen to minimize $\mathbb{E}[\ell(\rs{c}_{t})]$. Thus, 
\begin{multline}\label{eq:cbetter}
    \mathbb{E}[\ell(\rs{a}_t)|(\rs{a}^{t-1},\rs{d}^{t},\rs{u}^{t-1})=(a^{t-1},d^{t},u^{t-1})] \ge\\ \mathbb{E}[\ell(\rs{c}_t)|(\rs{a}^{t-1},\rs{d}^{t},\rs{u}^{t-1})=(a^{t-1},d^{t},u^{t-1})],
\end{multline} since choosing both $C_{a^{t-1}, d^{t}, u^{t-1}}$ and $D_{a^{t-1}, d^{t}, u^{t-1}}$ to be identity (i.e. choosing $\rs{c}_{t}=\rs{a}_{t}$) ensures that the decoder recovers $\rs{a}_{t}$ and that the $\{\rs{c}_{t}\}$ satisfy Assumption \ref{cs:constraint2}. Since the prefix constraint in Assumption \ref{cs:constraint2} applies for all realizations, we can lower bound 
$\mathbb{E}[\ell(\rs{c}_t)|\rs{a}^{t-1}=a^{t-1},  \rs{d}^{t}={d}^{t}, \rs{u}^{t-1}=u^{t-1}]$ using the standard Kraft-McMillan inequality based proof. For any realizations ($a^{t-1}$, $d^{t}$, $u^{t-1}$) and choice of code $C_{a^{t-1}, d^{t}, u^{t-1}}$, we have (cf. \cite[Theorem 5.3.1]{elemIT})
\begin{multline}\label{eq:conditionalentropy}
    \mathbb{E}[\ell(\rs{c}_t)|(\rs{a}^{t-1},\rs{d}^{t},\rs{u}^{t-1})=(a^{t-1},d^{t},u^{t-1})] \ge\\ H\left(\rs{a}_t|(\rs{a}^{t-1},\rs{d}^{t},\rs{u}^{t-1})=(a^{t-1},d^{t},u^{t-1})\right).
\end{multline} Taking the expectation of (\ref{eq:cbetter}) and (\ref{eq:conditionalentropy}) with respect to the joint measure $(\rs{a}^{t-1}, \rs{d}^{t}, \rs{u}^{t-1})$ over realizations allows us to proceed as in (\ref{eq:converse1}). We have
\begin{IEEEeqnarray}{rCl}\label{eq:converse2}
     \mathbb{E}[\ell(\rs{a}_t)] &\ge& \mathbb{E}[\ell(\rs{c}_t)] \label{eq:refcbetter}\\
      &\ge& H(\rs{a}_t|\rs{a}^{t-1},  \rs{d}^{t}, \rs{u}^{t-1}) \label{eq:refconditionalentropy}\\ &\ge& I(\rs{a}_{t};\rs{x}^{t}|\rs{a}^{t-1},\rs{d}^{t},\rs{u}^{t-1})\label{eq:defmi2},
\end{IEEEeqnarray} where (\ref{eq:refcbetter}) is by taking expectations over realizations in (\ref{eq:cbetter}), (\ref{eq:refconditionalentropy}) follows likewise from (\ref{eq:conditionalentropy}), and (\ref{eq:defmi2}) follows as in (\ref{eq:defmi1}).

 Summing the identical bounds in (\ref{eq:defmi1}) and (\ref{eq:defmi2}) over $t=\{0,\dots,T\}$ and applying the definition of causally conditioned DI from (\ref{eq:didef}) proves (\ref{eq:lb_lemm}). 
\end{proof}
Lemma \ref{lemm:didpi} lower bounds the DI in Lemma \ref{lemm:milb} by a DI amenable to the rate-distortion formulations in \cite{SDP_DI} and \cite{reviewerPaper}.
\begin{lemma}\label{lemm:didpi}
In the system model of Figure \ref{fig:ditharch}, we have
\begin{align}\label{eq:conversehypothesis}
I(\rs{x}^T\rightarrow\rs{a}^{T}||\rs{d}^{T},\rs{u}_{+}^{T-1})\ge I(\rs{x}^T\rightarrow\rs{u}^{T}). 
\end{align}
\end{lemma}
\begin{proof}
Let 
\begin{align}\label{eq:phi1def}
    \phi_t = I(\rs{x}^t;\rs{a}_t|\rs{a}^{t-1},\rs{d}^t,\rs{u}^{t-1})- I(\rs{x}^t;\rs{u}_t|\rs{u}^{t-1})
\end{align} and note that summing the $\phi_t$ at applying (\ref{eq:didef}) gives
\begin{align}\label{eq:sumprescope}
    I(\rs{x}^T\rightarrow\rs{a}^{T}||\rs{d}^{T},\rs{u}_{+}^{T-1})-I(\rs{x}^T\rightarrow\rs{u}^{T}) = \sum_{i=0}^{T}\phi_t. 
\end{align} We first demonstrate that 
\begin{multline}\label{eq:triplet}
        I(\rs{x}^t;\rs{a}_t|\rs{a}^{t-1},\rs{d}^t,\rs{u}^{t-1})  \\= I(\rs{x}^t;(\rs{a}_t,\rs{d}_{t},\rs{u}_t)|\rs{a}^{t-1},\rs{d}^{t-1},\rs{u}^{t-1}).
\end{multline} Via the chain rule, $ I(\rs{x}^t;(\rs{a}_t,\rs{d}_{t},\rs{u}_t)|\rs{a}^{t-1},\rs{d}^{t-1},\rs{u}^{t-1}) =$ $ I(\rs{x}^t;\rs{u}_{t}|\rs{a}^{t},\rs{d}^{t},\rs{u}^{t-1})+ I(\rs{x}^t;\rs{a}_{t}|\rs{a}^{t-1},\rs{d}^{t},\rs{u}^{t-1})+I(\rs{x}^t;\rs{d}_{t}|\rs{a}^{t-1},\rs{d}^{t-1},\rs{u}^{t-1})$. By (\ref{eq:ditherFactorization}), $\rs{d}_{t}$ is independent of $(\rs{a}^{t-1},\rs{d}^{t-1},\rs{u}^{t-1},\rs{x}^t)$ so  $I(\rs{x}^t;\rs{d}_{t}|\rs{a}^{t-1},\rs{d}^{t-1},\rs{u}^{t-1})=0$. Likewise,  (\ref{eq:ditherFactorization}) induces the Markov chain ($\rs{x}^{t}-(\rs{a}^t,\rs{d}^t,\rs{u}^{t-1})-\rs{u}_t$) (e.g. the control action at time $t$ is independent of the past states given the information at the decoder) and so $I(\rs{x}^t;\rs{u}_{t}|\rs{a}^{t},\rs{d}^{t},\rs{u}^{t-1})=0$. Substituting (\ref{eq:triplet}) in (\ref{eq:phi1def}) gives 
\begin{IEEEeqnarray}{rCl}
\phi_t &=&I(\rs{x}^t;(\rs{a}_t,\rs{d}_{t},\rs{u}_t)|\rs{a}^{t-1},\rs{d}^{t-1},\rs{u}^{t-1})- I(\rs{x}^t;\rs{u}_t|\rs{u}^{t-1})\nonumber
\\&=&I(\rs{x}^{t};(\rs{a}^t,\rs{d}^t)|\rs{u}^{t})-I(\rs{x}^{t};(\rs{a}^{t-1},\rs{d}^{t-1})|\rs{u}^{t-1})\label{eq:addnsub}\\&=&I(\rs{x}^{t};(\rs{a}^t,\rs{d}^t)|\rs{u}^{t})-I(\rs{x}^{t-1};(\rs{a}^{t-1},\rs{d}^{t-1})|\rs{u}^{t-1})\label{eq:lastmarkovconverse}
\end{IEEEeqnarray} Equality (\ref{eq:addnsub}) follows via expanding $ I(\rs{x}^t;(\rs{a}^t,\rs{d}^{t},\rs{u}_t)|\rs{u}^{t-1})$ via chain rule two different ways to show that
\begin{multline}\label{eq:addnsub2}
I(\rs{x}^t;(\rs{a}^{t-1},\rs{d}^{t-1})|\rs{u}^{t-1})+\\ I(\rs{x}^t;(\rs{a}_t,\rs{d}_{t},\rs{u}_t)|\rs{a}^{t-1},\rs{d}^{t-1},\rs{u}^{t-1}) =\\ I(\rs{x}^t;\rs{u}_t|\rs{u}^{t-1})+I(\rs{x}^t;(\rs{a}^t,\rs{d}^t)|\rs{u}^t),
\end{multline} and then adding right hand side and subtracting the left hand side of (\ref{eq:addnsub2}) from the preceding equation. For $t\ge 1$, (\ref{eq:lastmarkovconverse}) follows since by the chain rule $I(\rs{x}^{t};(\rs{a}^{t-1},\rs{d}^{t-1})|\rs{u}^{t-1})=I(\rs{x}^{t-1};(\rs{a}^{t-1},\rs{d}^{t-1})|\rs{u}^{t-1})+I(\rs{x}_{t};(\rs{a}^{t-1},\rs{d}^{t-1})|\rs{u}^{t-1},\rs{x}^{t-1})$. However, by the system model  (\ref{eq:ditherFactorization}) we have the Markov chain ($(\rs{a}^{t-1},\rs{d}^{t-1})-(\rs{x}^{t-1},\rs{u}^{t-1})-\rs{x}_t$) and so $I(\rs{x}_{t};(\rs{a}^{t-1},\rs{d}^{t-1})|\rs{u}^{t-1},\rs{x}^{t-1})=0$. When $t=0$, we have $I(\rs{x}^{t};(\rs{a}^{t-1},\rs{d}^{t-1})|\rs{u}^{t-1})=0$ and so we adopt the convention that $I(\rs{x}^{-1};(\rs{a}^{-1},\rs{d}^{-1})|\rs{u}^{-1})=0$. We can then apply (\ref{eq:lastmarkovconverse}) to telescope the sum in (\ref{eq:sumprescope}). This gives $\sum_{i=0}^{T}\phi_t=I(\rs{x}^{T};(\rs{a}^T,\rs{d}^T)|\rs{u}^{T})$. Substituting this into (\ref{eq:sumprescope}) and applying the non-negativity of MI proves (\ref{eq:conversehypothesis}).
\end{proof} Theorem \ref{thm:converse} is immediate upon combining Lemmas \ref{lemm:milb} and \ref{lemm:didpi}. 

We now briefly discuss how the lower bound in Theorem \ref{thm:converse} can be modified if we do not require that the codewords $\rs{a}_{t}$ conform to prefix constraints.  While prefix constraints can be useful in settings where multiple agents access a shared communication medium, they may be overly restrictive, in particular for the point-to-point case when the encoder and decoder share a common clock signal. The prefix constraints allowed us to use the entropy lower bounds in (\ref{eq:converse1shannon}) and (\ref{eq:conditionalentropy}) respectively, but lifting them permits a reduction in expected bitrate \cite{kostinaTradeoff} \cite{verduVariableLength}. Define the function $\theta(x):\mathbb{R}^{+}\rightarrow\mathbb{R}^{+}$ via $\theta(x) = x+(1+x)\log_2(1+x)-x\log_2(x)$. It can be shown that $\theta(x)$ is strictly increasing and concave. Thus, the inverse $\theta^{-1}:\mathbb{R}^{+}\rightarrow\mathbb{R}^{+}$ exists, is strictly increasing, and is convex. 

Assume that in the setting of Lemma \ref{lemm:milb}, we lift the prefix constraints and simply assume that $\rs{a}_{t}\in\{0,1\}^{*}$ for all $t$. Let $\rs{c}$ be a discrete random variable. By \cite[Equation (13)]{verduVariableLength}, if the support of $\rs{c}$ is mapped bijectively to the set of finite-length binary strings, the expected length of the string is greater than or equal to $\theta^{-1}(H(\rs{c}))$. Assume fixed sensor/encoder and decoder/controller policies.  As $\rs{a}_{t}$ is already a finite-length binary string, we have via \cite[Equation (13)]{verduVariableLength}
\begin{align}\label{eq:thetaDefVerdu}
    \mathbb{E}[\ell(\rs{a}_t)] \ge \theta^{-1}\left(H(\rs{a}_t)\right).
\end{align} Since $\theta^{-1}$ is increasing, we have by (\ref{eq:converse1condred})-(\ref{eq:defmi1}) that 
\begin{align}\label{eq:subsequentstep}
    \theta^{-1}\left(H(\rs{a}_t)\right) \ge \theta^{-1}\left(I(\rs{a}_{t};\rs{x}^{t}|\rs{a}^{t-1},\rs{d}^{t},\rs{u}^{t-1})\right).
\end{align} Combining (\ref{eq:thetaDefVerdu}) and (\ref{eq:subsequentstep}) gives the lower bound  $\mathbb{E}[\ell(\rs{a}_t)] \ge \theta^{-1}\left(I(\rs{a}_{t};\rs{x}^{t}|\rs{a}^{t-1},\rs{d}^{t},\rs{u}^{t-1})\right)$, which applies when  prefix constraints are relaxed.
Let $R_{t} = I(\rs{a}_{t};\rs{x}^{t}|\rs{a}^{t-1},\rs{d}^{t},\rs{u}^{t-1})$. By the convexity of $\theta^{-1}$ and Jensen's inequality, taking the time average gives
\begin{IEEEeqnarray}{rCl}
\sum_{i=0}^{T} \dfrac{\theta^{-1}(R_{t})}{T+1}\label{eq:citeserg}&\ge&  \theta^{-1}\left(\sum_{i=0}^{T}\dfrac{R_t}{T+1}\right)\label{eq:noprefixconversejensesns}\\ &\ge& \theta^{-1}\left(\dfrac{I(\rs{x}^T\rightarrow\rs{a}^{T}||\rs{d}^{T},\rs{u}_{+}^{T-1})}{T+1}\right)\label{eq:usedidef} \\&\ge& \theta^{-1}\left(\dfrac{I(\rs{x}^{T}\rightarrow\rs{u}^{T})}{T+1}\right)\label{eq:fromthedpi}
\end{IEEEeqnarray} where (\ref{eq:usedidef}) is the definition of DI and (\ref{eq:fromthedpi}) follows from the fact that $\theta^{-1}$ is increasing and the DI data processing inequality in Lemma \ref{lemm:didpi}.  Thus, if the prefix constraints in either Assumption \ref{cs:constraint1} or \ref{cs:constraint2} are relaxed, we have
 \begin{align} \label{eq:noasconverse}   \frac{1}{(T+1)}\sum_{i=0}^{T}\mathbb{E}[\ell(\rs{a}_t)] \ge  \theta^{-1}\left(\frac{I(\rs{x}^{T}\rightarrow\rs{u}^{T})}{T+1}\right).
 \end{align} The next section motivates a rate-distortion optimization seeking to minimize the time-average DI lower bound from (\ref{eq:converse_thm}) subject to constraints on control performance. As $\theta^{-1}$ is increasing and convex, (\ref{eq:noasconverse}) makes this meaningful even when prefix constraints are relaxed. 
 
\subsection{Rate Distortion Formulation}\label{ssec:rdf}
  We reexamine the optimization proposed in (\ref{eq:codewordLenghtOptimization}) in light of the converse result obtained in Theorem \ref{thm:converse}. We assume a time-invariant plant, e.g. $A_{t}=A$ and $B_{t}=B$, where (A,B) are stabilizable to ensure finite-control cost is attainable. Define the infinite-horizon generalization of (\ref{eq:codewordLenghtOptimization}) via 
  \begin{align}\nonumber
    \mathcal{L}^{*} &=& \left\{\begin{aligned}
& \underset{\mathbb{P}_\mathrm{E}, \mathbb{P}_{\mathrm{C}}}{\inf}  \underset{T\rightarrow \infty}{\lim\sup}\text{ }\frac{1}{T+1}\sum\nolimits_{t=0}^{T}\mathbb{E}[\ell(\rvec{a}_{t})] \\ &\text{s.t. }   \underset{T\rightarrow \infty}{\lim\sup}\frac{\sum\nolimits_{t=0}^{T}\mathbb{E}[\lVert \rvec{x}_{t+1} \rVert_{{Q}}^{2} +\lVert \rvec{u}_{t} \rVert_{{R}}^{2}]}{T+1} \le \gamma
\end{aligned}.\right.
\end{align} Defining  $\{\mathbb{P}_\mathrm{E}, \mathbb{P}_{\mathrm{C}}\}$ via (\ref{eq:encdithpol}) and (\ref{eq:contdithpol}), Theorem \ref{thm:converse} gives
\begin{align}\label{eq:optimization2}
    \mathcal{L}^*\ge \left\{\begin{aligned}
& \underset{\mathbb{P}_\mathrm{E}, \mathbb{P}_{\mathrm{C}}}{\inf}  \underset{T\rightarrow \infty}{\lim\sup}\text{ }\frac{1}{T+1}I(\rs{x}^T\rightarrow\rs{u}^{T}) \\ &\text{s.t. }   \underset{T\rightarrow \infty}{\lim\sup}\frac{\sum\nolimits_{t=0}^{T}\mathbb{E}[\lVert \rvec{x}_{t+1} \rVert_{{Q}}^{2} +\lVert \rvec{u}_{t} \rVert_{{R}}^{2}]}{T+1} \le \gamma
\end{aligned}. \right. 
\end{align} We can relax the optimization problem in (\ref{eq:optimization2}) by expanding these policy spaces. Note that both the control and communication costs involve only $\{\rs{x}_{i},\rs{u}_{i}\}$. By the system model, we have the factorization $ \mathbb{P}_{\rs{x}^{T},\rs{u}^{T}}[\rs{x}^{T},\rs{u}^{T}] =\prod_{t=0}^{T} \mathbb{P}_{\rs{x}_{t}|\rs{u}_{t-1},\rs{x}_{t-1}}[\rs{x}_{t}|\rs{u}_{t-1},\rs{x}_{t-1}]\mathbb{P}_{\rs{u}_{t}|\rs{u}^{t-1},\rs{x}^{t}}[\rs{u}_{t}|\rs{u}^{t-1},\rs{x}^{t}]$. In a sense, this follows from  causality. The kernels
$\{\mathbb{P}_{\rs{x}_{t}|\rs{u}_{t-1},\rs{x}_{t-1}}[\rs{x}_{t}|\rs{u}_{t-1},\rs{x}_{t-1}]\}$ are time-invariant and fixed by the plant model. Meanwhile, the kernels of the form
$\mathbb{P}_{\rs{u}_{t}|\rs{u}^{t-1},\rs{x}^{t}}[\rs{u}_{t}|\rs{u}^{t-1},\rs{x}^{t}]$ are induced by the encoder and controller policies. Let $\mathbb{P}_{\rs{u}||\rs{x}}$ denote the set of all sequences of Borel-measurable kernels of the form $\{\mathbb{P}_{\rs{u}_{t}|\rs{u}^{t-1},\rs{x}^{t}}[\rs{u}_{t}|\rs{u}^{t-1},\rs{x}^{t}]\}$. Note that $\mathbb{P}_{\rs{u}||\rs{x}}$ contains all kernels that can  be induced by encoder and controller policies satisfying (\ref{eq:encdithpol}) and (\ref{eq:contdithpol}). Consider the optimization 
\begin{align}\label{eq:rdf}
    \mathcal{R} = \left\{\begin{aligned}
& \underset{\mathbb{P}_{\rs{u}||\rs{x}}}{\inf}  \underset{T\rightarrow \infty}{\lim\sup}\text{ }\frac{1}{T+1}I(\rs{x}^T\rightarrow\rs{u}^{T}) \\ &\text{s.t. }     \underset{T\rightarrow \infty}{\lim\sup}\frac{\sum\nolimits_{t=0}^{T}\mathbb{E}[\lVert \rvec{x}_{t+1} \rVert_{{Q}}^{2} +\lVert \rvec{u}_{t} \rVert_{{R}}^{2}]}{T+1} \le \gamma
\end{aligned}.\right.
\end{align} 
As the domain of optimization in (\ref{eq:rdf}) is expanded with respect to that on the right-hand side of (\ref{eq:optimization2}), we have $ \mathcal{R}\le \mathcal{L}^{*}$.

The optimization in (\ref{eq:rdf}) is the subject of \cite{SDP_DI}. While, a priori, the optimization in (\ref{eq:rdf}) is over an infinite-dimensional policy space, \cite{SDP_DI} demonstrated that the minimum in (\ref{eq:rdf}) could be computed by a finite dimensional log-determinant optimization. Let $S$ be a stabilizing solution to the discrete algebraic Riccati equation $A^{\mathrm{T}}SA-S-A^{\mathrm{T}}SB(B^{\mathrm{T}}SB+R)^{-1}B^{\mathrm{T}}SA+Q = 0$, let $K=-(B^{\mathrm{T}}SB+R)^{-1}B^{\mathrm{T}}SA$, and let $\Theta = K^{\mathrm{T}}(B^{\mathrm{T}}SB+R)K$. It can be shown (cf. \cite[Section IV.B]{SDP_DI}) that the optimization in (\ref{eq:rdf}) is equivalent to
\begin{align}\label{eq:threestageRDF}
    \mathcal{R} &=& \left\{\begin{aligned}
& \underset{\substack{P,\Pi, \in\mathbb{R}^{m\times m}\\P,\Pi\succeq 0} }{\inf} \frac{1}{2}(-\log_2{\det{\Pi}}+\log_2{\det{W}} )\\ &\text{ }\text{s.t. }  \mathrm{Tr}(\Theta P)+\mathrm{Tr}(WS)\le \gamma \\ &\text{ }\text{ } P\preceq APA^\mathrm{T}+W\\&\text{ }\text{ }\begin{bmatrix}P-\Pi & PA^{\mathrm{T}} \\ AP & APA^{\mathrm{T}}+W \end{bmatrix}\succeq 0 
\end{aligned}.\right.
\end{align} The optimization (\ref{eq:threestageRDF}) is convex and amenable to solution via standard libraries\cite{SDP_DI}. 
\section{Conclusion}
In \cite{silvaFirst}, \cite{tanakaISIT} prefix-free coding schemes that conform to Assumption \ref{cs:constraint2} were shown to nearly achieve the communication cost for systems with a shared uniform dither sequence at the encoder and decoder. Likewise, bounds on the quantizer output entropy from \cite{kostinaTradeoff} (not assuming dithering) can be used to demonstrate the existence of a scheme conforming to Assumption \ref{cs:constraint1} that approximately achieves the lower bound in the high communication rate/strict control cost regime. 

The prefix constraints discussed in this work apply ``at time $t$" in the sense that both  definitions allow different prefix-free codebooks to be used at every time $t$. For example, under Assumption \ref{cs:constraint1} there is nothing preventing a codeword at time $t+1$ from being a prefix of some codeword at time $t$. Likewise, proofs of the achievability results in \cite{tanakaISIT}, \cite{kostinaTradeoff} imply the codebook is time-varying. Enforcing a time-invariant prefix constraint could enable more computationally efficient communication resource sharing in a network scenario; the end of the codeword can be detected by comparing received transmissions against a time-invariant list. Likewise, a fully time-invariant source coding scheme (e.g. where both the quantizer and the the mapping from quantizations to prefix-free codewords is itself time-invariant) eliminates the need to update codebooks at every timestep and reduces the computational complexity of both encoders and decoders. Bounds for time-invariant codecs are an opportunity for future work. 

\bibliographystyle{IEEEtran}
\bibliography{refs.bib}

\end{document}